\documentclass[12pt,twoside]{article}

\usepackage[a4paper]{geometry}
\makeatletter
\renewcommand\title[1]{\gdef\@title{\reset@font\Large\bfseries #1}}
\renewcommand\section{\@startsection {section}{1}{\z@}%
                                   {-3.5ex \@plus -1ex \@minus -.2ex}%
                                   {2.3ex \@plus.2ex}%
                                   {\normalfont\large\bfseries}}
\renewcommand\subsection{\@startsection{subsection}{2}{\z@}%
                                     {-3ex\@plus -1ex \@minus -.2ex}%
                                     {1.5ex \@plus .2ex}%
                                     {\normalfont\normalsize\bfseries}}
\renewcommand\subsubsection{\@startsection{subsubsection}{3}{\z@}%
                                     {-2.5ex\@plus -1ex \@minus -.2ex}%
                                     {1.5ex \@plus .2ex}%
                                     {\normalfont\normalsize\bfseries}}

\def\@runningauthor{}\newcommand{\runningauthor}[1]{\def\runningauthor{#1}}
\def\@runningtitle{}\newcommand{\runningtitle}[1]{\def\runningtitle{#1}}

\renewcommand{\ps@plain}{%
\renewcommand{\@evenfoot}{\footnotesize Sequences and Their Applications (SETA) 2016\hfill\thepage}
\renewcommand{\@oddfoot}{\footnotesize Sequences and Their Applications (SETA) 2016 \hfill\thepage}
\renewcommand{\@evenhead}{\footnotesize\scshape \hfill\runningauthor\hfill}
\renewcommand{\@oddhead}{\footnotesize\scshape \hfill\runningtitle\hfill}}
\g@addto@macro\bfseries{\boldmath}

\makeatother

\usepackage{amsthm,amsmath,amssymb}

\usepackage{graphicx}

\usepackage[colorlinks=true,citecolor=black,linkcolor=black,urlcolor=blue]{hyperref}
\usepackage{bm}
\usepackage{url}
\usepackage{breakurl}
\theoremstyle{plain}
\newtheorem{theorem}{Theorem}
\newtheorem{lemma}[theorem]{Lemma}
\newtheorem{corollary}[theorem]{Corollary}

\theoremstyle{definition}

\newtheorem{example}[theorem]{Example}

\newtheorem{construction}[theorem]{Construction}
\theoremstyle{remark}
\newtheorem{remark}[theorem]{Remark}

\title{A New Method to Construct Gloay Complementary Set by Paraunitary Matrices and Hadamard Matrices}

\runningtitle{Constructtion of Gloay Complementary Set by Paraunitary Matrices and Hadamard Matrices}

\author{Zilong Wang\thanks{The authors are supported by funding organisation xyz.}, Gaofei Wu, Dongxu Ma \\
\small State Key Laboratory of Integrated Service Networks,\\[-0.8ex]
\small Xidian University\\[-0.8ex]
\small Xi'an, China\\
\small\tt zlwang@xidian.edu.cn, gaofei\_wu@qq.com, xidiandongxuma@foxmail.com\\
}

\runningauthor{Zilong Wang, Gaofei Wu, Dongxu Ma}

\date{}

\begin{document}

\maketitle

\thispagestyle{empty}

\begin{abstract}
Golay complementary sequences have been put a high value on the
applications in orthogonal frequency-division multiplexing (OFDM)
systems since its good peak-to-mean envelope power ratio(PMEPR)
properties. However, with the increase of the code length, the code
rate of the standard Golay sequences suffer a dramatic decline. Even
though a lot of efforts have been paid to solve the code rate
problem  for OFDM application, how to
construct large classes of sequences with low PMEPR is still
difficult and open now. In this paper, we propose a new method to
construct $q$-ary  Golay complementary set of size $N$ and length
$N^n$ by $N\times N$ Hadamard Matrices where $n$ is arbitrary and
$N$ is a power of 2. Every item of the constructed sequences can be
presented as the product of the specific entries of the
Hadamard Matrices.  The previous works in \cite{BudIT} can be
regarded as a special case of the constructions in this paper and we
also obtained new quaternary Golay sets never reported in the
literature.
\end{abstract}


\section{Introduction}
In the fundamental paper written by M. Golay\cite{Golay1961},
binary complementary sequences were first introduced in the context
of infrared spectrometry. A complementary pair has the useful
property that the aperiodic autocorrelations of the two sequences
sum up to zero for each nonzero shift. Golay complementary pairs and
sequences have been found many applications in the fields of science and
engineering, especially in wireless communication. One of
the most important applications is in orthogonal frequency-division
multiplexing (OFDM) \cite{Popovic91}\cite{Lim2009}. The good property of $q$-ary Golay sequences, that their peak-to-mean
envelope power ratio (PMPPR) is at most 2, makes them appropriate to use as codewords
in OFDM system.

In \cite{Golay1961}, Golay proposed recursive concatenation and interleaving algorithms to construct
binary Golay complementary pairs of length $2^n$. Binary Golay complementary
pairs were generalized to complementary sets in \cite{Tseng1972}, and to polyphase
sequences in \cite{Sivaswamy1978}. In 1999, Davis and Jedwab's milestone work showed that known $q$-ary sequences in Golay pair fill up
specific second-order cosets of the generalized first-order
Reed-Muller codes \cite{DavisJedwab99}.Such sequences are so-called Golay-Davis-Jedwab (GDJ) sequences or standard
Golay sequences. However, with the increase of the code length, the
code rate of the standard Golay sequences suffers a dramatic decline \cite{DavisJedwab99}.

To solve the code rate problem of the standard
Golay sequences for OFDM application, a lot of efforts have been paid in the literature where PMEPR is bounded by a
finite value $c\geqslant 2$, such as finding more non-standard Golay sequence\cite{Li2005}\cite{fiedler2006}\cite{fiedler2},
design Golay complementary set  \cite{paterson}\cite{schmidt07},design near complementary sequences\cite{schmidt06}, design sequences with low PMEPR over QAM constellation\cite{Tarokh2001}\cite{Tarokh2003}\cite{Wang2009}, and design Golay complementary sequences over QAM constellation\cite{Chong2003}\cite{Li2010}\cite{Tarokh2003}\cite{Liu2013}. However, these coding method did not solve the code rate problem in OFDM system. For this reason, how to construct large classes of sequences with low PMEPR
is critical and still open now.

All the above Boolean-function-based constructions are proved by the
property that the summation of the aperiodic autocorrelations of the
two or more sequences equals zero for each nonzero shift. There
exists related matrix-based approaches to complementary sequence
design in the literature. For instance, the complete complementary
code approach of \cite{Suehiro} and the paraunitary matrix approach
to filter banks for complementary sequences, as discussed in
\cite{BudSpas}\cite{BudQAM}. In particular, Bud\v{i}sin and
Spasojevi$\acute{c}$ proposed a new method to construct Golay
complementary pairs by paraunitary matrices over both PSK and QAM
constellations in \cite{BudIT}. By this method,  not only all the
standard Golay pair over PSK constellation and known Golay pair over
QAM constellation can be explained, but also some new QAM complementary pairs can be constructed. Here the word
`paraunitary' refers to a matrix polynomial in $z$ that is unitary
for all $z$ on unit circle.

In this paper, inspired by the paraunitary-matrix-based method in
\cite{BudIT}, we propose a new method to construct $q$-ary  Golay
complementary set of size $N$ and length $N^n$ by $N\times N$
Hadamard matrices where $n$ is arbitrary and $N$ is a power of 2.
Every item of the constructed sequences can be presented as the
product of the specific entries of the Hadamard matrices. To
make the results easier to understand for readers, we only show the
proof for $N=2^2=4$, and the case for othe is $N$ similar to $N=4$. Bud\v{i}sin
and  Spasojevi$\acute{c}$'s algorithm for $q$-ary Golay pair in
\cite{BudIT} can be regarded as a special case of the construction  in this
paper by setting $N=2$ and selecting specific Hadamard matrices. We
also show new quaternary Golay pair of size 4 and degree 3 never
reported in the literature can be obtained by the method presented
in this paper.





\section{Preliminaries}

\subsection{Peak Power Control in OFDM and Golay Set}
Let $\textbf{a}=(a_0,a_1,\cdots,a_{L-1})$ be a $q$-ary sequence of
length $L$ where $a_i$ is a $q$th roots of
unity. In an  OFDM system with $L$ subcarriers, the transmitted signal by
employing sequence $\textbf{a}$ can be modeled as the real part of
$$s_\textbf{a}(t)=\sum_{i=0}^{L-1}{a_ie^{2\pi\sqrt{-1}(f_0+i\bigtriangleup
f)t}},\  t\in [0,\frac{1}{\bigtriangleup f}),$$ where
$\bigtriangleup f$ is the frequency separation between adjacent
subcarrier pairs and  $f_0$ is the base frequency.

The sequence $\textbf{a}$ can be associated with the polynomial
\begin{equation}\label{ploynomial}
A(z)=\sum_{i=0}^{L-1}{{a_i}z^i}.
\end{equation}
in indeterminate $z$. We use $A(z)$ instead of sequence $\textbf{a}$ in the rest of the paper for convenience.

By restricting $z$ to lie on the unit circle in the complex plane,
i.e.,$z=e^{2\pi\sqrt{-1}\bigtriangleup
ft}$. We have
$$|s_\textbf{a}(t)|=|A(e^{2\pi\sqrt{-1}\bigtriangleup ft})|.$$
Then the {\em instantaneous envelope power} of the
transmitted signal is determined by $A(z)\overline{A}(z^{-1})$
where the conjugate polynomial $\overline{A}(z^{-1})=\sum_{i=0}^{L-1}{{\overline{a_i}}z^{-i}}$, and the  {\em
peak-to-mean envelope power ratio} (PMEPR)  is determined  by
\begin{equation}\label{pmepr}
\mbox{PMEPR}(\textbf{a})=\frac{1}{L}\sup_{|z|=1}A(z)\overline{A}(z^{-1}).
\end{equation}

If a set of sequences $A_i(z)$ for $1\leqslant i\leqslant N$ satisfy
\begin{equation}\label{Golayset}
\sum_{i=1}^N A_i(z\overline{)A_i}(z^{-1})=LN,
\end{equation}
$\{A_i(z): 1\leqslant i\leqslant N\}$ is called a Golay set of size $N$. It is straightforward that
the PMEPR of each sequence $A_i(z)$ in Golay set is upper bounded by $N$ by definition in (\ref{pmepr}).

\subsection{Butson Type Hadamard matrix}
A complex Hadamard matrix is a complex $N\times N$ matrix $H$
satisfying $|H_{ij}|=1$ ($i,j=1,2,\cdots, N$) and
$HH^{\dagger}=N\cdot I_{N}$, where $ H^{\dagger}$  denotes the
Hermitian transpose of \(H\) and \({I}_N\) is the identity matrix. A
complex Hadamard matrix of size \(N\) is called of Butson type
\(H(q,N)\)  \cite{Butson1962} if all the entries of \(H\) are $q$th roots of
unity. We are only interested in Butson type Hadamard matrices in the rest of the  paper.

Two Butson type $(q,N)$ Hadamard matrices  $H_1$ and $H_2$ are called
equivalent,  written as \(H_1\simeq H_2\), if there exist diagonal
unitary matrices \(D_1\), \(D_2\) where each diagonal entry of $D_i$ is $q$th root of
unity and permutation matrices \(P_1\),
\(P_2\) such that:

\begin{equation}\label{Hadamard}
H_1 = D_1 P_1 H_2 P_2 D_2.
\end{equation}



For example, all binary Hadamard  matrices in  \(H(2,4)\) are equivalent to the $4\times 4$ Hadamard matrix:
$$
\left ( \begin{array}{cccc}
1 & 1 &1&1\\1 & -1&1&-1\\
 1 &1&-1&-1\\
  1 &-1&-1& 1
\end{array}  \right ),
$$
so every (2,4)  Hadamard  matrices can be generated by the above matrix.
And all quaternary  Hadamard  matrices in  \(H(4,4)\) are equivalent to one of  the following
 $4\times 4$ Hadamard matrix:
$$
\left ( \begin{array}{cccc}
1 & 1 &1&1\\1 & -1&1&-1\\
 1 &1&-1&-1\\
  1 &-1&-1& 1
\end{array}  \right ),
~and~
\left ( \begin{array}{cccc}
1 & 1 &1&1\\1 &i &-1&-i\\
 1 &-1&1&-1\\
  1 &-i&-1& i
\end{array}  \right ),
$$
so  every (4,4)  Hadamard  matrices can be generated by one of the above matrices by using formula (\ref{Hadamard}) .

\subsection{Paraunitary  Matrix and Golay Set}
An $N\times N$ complex Matrix  is called unitary if $MM^{\dagger}=I_{N}$.
An $N\times N$ polynomial matrix $M(z)$, where every entry $M_{ij}(z)$  is a polynomial  in indeterminate $z$ with complex coefficients, is called paraunitary if
\begin{equation}\label{Paraunitary}
M(z)M^{\dagger}(z^{-1})=c\cdot I_{N}
\end{equation}
for real constant $c$. We are only interested in the paraunitary matrix $M(z)$ where all the coefficients of every entry $M_{ij}(z)$ are $q$th roots of
unity in the rest of the  paper.

In the next section, we present a new method to design $N\times N$  paraunitary  matrix $M(z)$ satisfying
\begin{equation}\label{Paraunitary2}
M(z)M^{\dagger}(z^{-1})=L\cdot I_{N},
\end{equation}
where every entry $M_{ij}(z)$ is a polynomial of degree $L-1$ with $q$th root coefficients, i.e., every entry corresponds to a $q$-ary sequence of length $L$.
From the definition of the Golay set, the sequences in every row or every column of the matrix $M(z)$ satisfying (\ref{Paraunitary2}) form a $q$-ary Golay set of size $N$ and length $L$.

\section{New Construction}
In this section, by using paraunitary matrix and  Butson type $(q,N)$ Hadamard matrices,  we propose a new method to construct Golay set.

\begin{construction}\label{cons:main}
Let $N$ be a power of 2, $n$  a non-negative integer, and $H^{(i)}$ an arbitrary Butson type $(q,N)$ Hadamard matrix for $0\leqslant i\leqslant n$ if it exists .
Let $D(z)$ be a diagonal paraunitary matrix with the form $D(z)=diag\{{1, z, z^2,\cdots, z^{N-1}}\}$, and $\pi$ a permutation
of numbers $\{0,1,2,\cdots,n-1\}$. Define $N\times N$  matrix $M^{\{n\}}(z)$ as follows:
\begin{small}
\begin{eqnarray}
M^{\{n\}}(z)&=&H^{\{0\}}\cdot (D(z))^{N^{\pi(0)}}\cdot H^{\{1\}}\cdots
(D(z))^{N^{\pi(n-1)}}\cdot H^{{\{n\}}}\nonumber\\
&=&H^{\{0\}}\left\{\prod_{t=1}^{n}\left((D(z))^{N^{\pi(t-1)}}\cdot H^{\{t\}}\right)\right\}\nonumber\\
&=&H^{\{0\}}\cdot\prod_{t=1}^{n}\left(\begin{bmatrix}1&0&0&0&\cdots&0\\0&z^{
{N^{\pi(t)}}}&0&0&\cdots&0
\\0&0&z^{2\cdot N^{\pi(t)}}&0&\cdots&0\\0&0&0&z^{3\cdot N^{\pi(t)}}&\cdots&0\\\vdots&\vdots&\vdots&\vdots&\vdots&\vdots\\0&0&0& 0&\cdots&z^{(N-1)\cdot N^{\pi(t)}}\end{bmatrix}
\cdot H^{\{t\}}\right).
\end{eqnarray}
\end{small}
\end{construction}

\begin{theorem}
The sequences in every row or every column of the matrix $M^{\{n\}}(z)$ defined above form a $q$-ary Golay set of size $N$ and length $N^n$.
\end{theorem}

To verify the Theorem 2, we need to prove the following properties by Section 2.3.
\begin{itemize}
\item{(1)} $M^{\{n\}}(z)$ is a paraunitary  matrix satisfying $M^{\{n\}}(z)M^{{\{n\}}\dagger}(z^{-1})=N^{n+1}\cdot I_{N}$
\item{(2)} The degree of every entry of $M_{ij}^{\{n\}}(z)$ is $N^n-1$
\item{(3)} All the coefficients of every entry  $M_{ij}^{\{n\}}(z)$ are $q$th roots of
unity
\end{itemize}

The first property is very easy to check, since $H^{(i)}H^{(i)\dagger}=N\cdot I_{N}$ for $0\leqslant i\leqslant n$, and $D(z)D(z)^{\dagger}=I_{N}$.
We prove the second and third properties by verifying every coefficient of $M_{ij}^{\{n\}}(z)$  presented by the product of the entries of
the involved Butson type Hadamard matrices.

For $N=2$, the proof for Theorem 2 can be easily derived from Bud\v{i}sin and  Spasojevi$\acute{c}$'s algorithm in \cite{BudIT},
which can generate all $q$-ary standard Golay pairs. We will show the proof for the case $N=4$ in the next section, in which new Golay set of size $4$ can be obtained.
For general $N$, the proof is similar to the case $N=4$.

\section{Proof for $N=4$}

 Suppose that $N=4$ and $N^n=4^n=L$, $m$ is an  integer and
$m\in[0,L-1]$. Note that the binary expansion of  $m$ is:
\begin{equation}
\label{equ:dec}
m=\sum^{n-1}_{k=0}(\delta_{2k}(m)2^{2k}+\delta_{2k+1}(m)2^{2k+1}).
\end{equation}
\begin{lemma}
\label{lem:gong} Let $L=4^n$, where $n$ is a non-negative integer.
Then we have
\begin{equation}
\prod^{n-1}_{k=0}(F_k(0,0)+F_k(1,0)+F_k(0,1)+F_k(1,1))=\sum^{L-1}_{m=0}\left\{\prod^{n-1}_{k=0}
F_k(\delta_{2k}(m),\delta_{2k+1}(m))\right\},
\end{equation}
where $F_k(x_0,x_1)$ are  matrix complex Boolean functions,
i.e.
$F_k(x_0,x_1)=M_1\cdot\overline{x_0}\cdot\overline{x_1}+M_2\cdot
x_0\cdot\overline{x_1}+M_3\cdot\overline{x_0}\cdot x_1+M_4\cdot
x_0\cdot x_1$ where $M_1,M_2,M_3,M_4$ are arbitrary complex
matrices and $\overline{0}=1, \overline{1}=0$.
\end{lemma}
\begin{proof}
\begin{eqnarray}
&&\prod^{n-1}_{k=0}(F_k(0,0)+F_k(1,0)+F_k(0,1)+F_k(1,1))\nonumber\\
&=&(F_0(0,0)+F_0(1,0)+F_0(0,1)+F_0(1,1))\times\nonumber\\
&&(F_1(0,0)+F_1(1,0)+F_1(0,1)+F_1(1,1))\times\cdots\times\nonumber\\
&&(F_{n-2}(0,0)+F_{n-2}(1,0)+F_{n-2}(0,1)+F_{n-2}(1,1))\times \nonumber\\
&&(F_{n-1}(0,0)+F_{n-1}(1,0)+F_{n-1}(0,1)+F_{n-1}(1,1))\nonumber\\
&=&F_{n-1}(0,0)\cdots F_{1}(0,0)F_{0}(0,0)+F_{n-1}(0,0)\cdots
F_{1}(0,0)F_{0}(0,1)+\nonumber\\
&&\cdots+ F_{n-1}(1,1)\cdots F_{1}(1,1)F_{0}(1,1)\nonumber\\
&=&\sum^1_{\delta_0=0}\sum^1_{\delta_1=0}\cdots\sum^1_{\delta_{2n-1}=0}\left\{\prod^{n-1}_{k=0}F_k(\delta_{2k},\delta_{2k+1})\right\}\nonumber\\
&=&\sum^{L-1}_{m=0}\left\{\prod^{n-1}_{k=0}F_k(\delta_{2k}(m),\delta_{2k+1}(m))\right\}.
\end{eqnarray}
\end{proof}

\begin{corollary}
Let $\pi$ be a permutation
of numbers $\{0,1,2,\cdots,n-1\}$, then the equation in
Lemma~\ref{lem:gong} can be rewritten as:
\begin{equation}
\prod^{n-1}_{t=0}(F_{\pi(t)}(0,0)+F_{\pi(t)}(1,0)+F_{\pi(t)}(0,1)+F_{\pi(t)}(1,1))=\sum^{L-1}_{m=0}\left\{\prod^{n-1}_{t=0}
F_{\pi(t)}(\delta_{2\pi(t)}(m),\delta_{2\pi(t)+1}(m))\right\}.
\end{equation}
\end{corollary}

Let $$M^{\{n\}}(z)=\sum^{L-1}_{m=0}M^{\{n\}}(m)\cdot z^m,$$
where $M^{\{n\}}(m)$ is the coefficient matrix of the polynomial matrix $M^{\{n\}}(Z)$ of item $z^m$. We show that
every entry of $M^{\{n\}}(m)$ is $q$th root of unity in the following theorem, which complete the proof of Theorem 2.

\begin{theorem}
\label{thm:M}
Let $N=4$ in Construction \ref{cons:main}. The coefficient matrix of the polynomial matrix $M^{\{n\}}(z)$ of item $z^m$ can be presented by

\begin{eqnarray}
\label{equ:M}
&&M^{\{n\}}(m)=H^{\{0\}}\cdot\prod^{n}_{t=1}\left(A_t\cdot
H^{\{t\}}\right),
\end{eqnarray}
where
\begin{eqnarray}
A_{t+1}=\begin{bmatrix}\overline{\delta_{2\pi(t)}(m)}\cdot
\overline{\delta_{2\pi(t)+1}(m)}&0&0&0\\
0&\delta_{2\pi(t)}(m)\cdot\overline{\delta_{2\pi(t)+1}(m)}&0&0\\
0&0&\overline{\delta_{2\pi(t)}(m)}\cdot\delta_{2\pi(t)+1}(m)&0\\
0&0&0&\delta_{2\pi(t)}(m)\cdot\delta_{2\pi(t)+1}(m)\end{bmatrix}\nonumber
\end{eqnarray}

\begin{remark}
Only one entry of matrix $A_t$ is 1, and the others are all 0.
\end{remark}

\end{theorem}
\begin{proof}
It is sufficient to prove  that the univariate polynomial with coefficients in
(\ref{equ:M}) equals to the expression in Construction \ref{cons:main} with $N=4$, i.e.,
$$\sum^{L-1}_{m=0} H^{\{0\}}\cdot\prod^{n}_{t=1}\left(A_t\cdot
H^{\{t\}}\right)z^m=M^{\{n\}}(z).$$

Here we substitute $m$ with its binary expansion:
\begin{eqnarray*}
m&=&\sum^{n-1}_{k=0}(\delta_{2k}2^{2k}+\delta_{2k+1}2^{2k+1})\\
&=&\sum^{n-1}_{t=0}(\delta_{2\pi(t)}2^{2\pi(t)}+\delta_{2\pi(t)+1}2^{2\pi(t)+1}).
\end{eqnarray*}
Then we have
\begin{footnotesize}
\begin{eqnarray}
\label{equ:main}
& &\sum^{L-1}_{m=0} H^{\{0\}}\cdot\prod^{n}_{t=1}\left(A_t\cdot
H^{\{t\}}\right)z^m \nonumber\\
&=&H^{\{0\}}\cdot\sum^{L-1}_{m=0}\left\{\left\{\prod_{t=1}^{n}\left(A_t
\cdot H^{\{t\}}\right)\right\}\cdot
z^{\sum^{n-1}_{t=0}(\delta_{2\pi(t)}2^{2\pi(t)}+\delta_{2\pi(t)+1}2^{2\pi(t)+1})}\right\}\nonumber\\
&=&H^{\{0\}}\cdot\sum^{L-1}_{m=0}\left\{\left\{\prod_{t=1}^{n}\left(A_t
\cdot H^{\{t\}}\right)\right\}\cdot
\prod^{n-1}_{t=0}z^{\delta_{2\pi(t)}2^{2\pi(t)}+\delta_{2\pi(t)+1}2^{2\pi(t)+1}}\right\}\nonumber\\
&=&H^{\{0\}}\cdot\sum^{L-1}_{m=0}\left\{\prod_{t=0}^{n-1}\left(A_{t+1}
\cdot
z^{\delta_{2\pi(t)}2^{2\pi(t)}+\delta_{2\pi(t)+1}2^{2\pi(t)+1}}H^{\{t+1\}}\right)\right\}
\end{eqnarray}
\end{footnotesize}
%
%
%
Note that
$$\overline{\delta_{2\pi(t)}(m)}\cdot
\overline{\delta_{2\pi(t)+1}(m)}\cdot
z^{\delta_{2\pi(t)}2^{2\pi(t)}+\delta_{2\pi(t)+1}2^{2\pi(t)+1}}\cdot
H^{\{t+1\}}=\overline{\delta_{2\pi(t)}(m)}\cdot
\overline{\delta_{2\pi(t)+1}(m)}\cdot H^{\{t+1\}},$$
$$\delta_{2\pi(t)}(m)\cdot\overline{\delta_{2\pi(t)+1}(m)}\cdot
z^{\delta_{2\pi(t)}2^{2\pi(t)}+\delta_{2\pi(t)+1}2^{2\pi(t)+1}}\cdot
H^{\{t+1\}}=\delta_{2\pi(t)}(m)\cdot\overline{\delta_{2\pi(t)+1}(m)}\cdot
z^{2^{2\pi(t)}}\cdot H^{\{t+1\}},$$

$$\overline{\delta_{2\pi(t)}(m)}\cdot\delta_{2\pi(t)+1}(m)\cdot
z^{\delta_{2\pi(t)}2^{2\pi(t)}+\delta_{2\pi(t)+1}2^{2\pi(t)+1}}\cdot
H^{\{t+1\}}=\overline{\delta_{2\pi(t)}(m)}\cdot\delta_{2\pi(t)+1}(m)\cdot
z^{2^{2\pi(t)+1}}\cdot H^{\{t+1\}},$$
and
$$\delta_{2\pi(t)}(m)\cdot\delta_{2\pi(t)+1}(m)\cdot
z^{\delta_{2\pi(t)}2^{2\pi(t)}+\delta_{2\pi(t)+1}2^{2\pi(t)+1}}\cdot
H^{\{t+1\}}=\delta_{2\pi(t)}(m)\cdot\delta_{2\pi(t)+1}(m)\cdot
z^{2^{2\pi(t)}+2^{2\pi(t)+1}}\cdot H^{\{t+1\}},$$
we have
\begin{footnotesize}
\begin{eqnarray}
&&\prod_{t=0}^{n-1}\left(A_t \cdot
z^{\delta_{2\pi(t)}2^{2\pi(t)}+\delta_{2\pi(t)+1}2^{2\pi(t)+1}}\right)
H^{\{t+1\}}\nonumber\\
&=&\prod_{t=0}^{n-1}\left(\begin{bmatrix}1&0&0&0\\0&0&0&0\\0&0&0&0\\0&0&0&0\end{bmatrix}\overline{\delta_{2\pi(t)}(m)}\cdot
\overline{\delta_{2\pi(t)+1}(m)}\right.
+\begin{bmatrix}0&0&0&0\\0&1&0&0\\0&0&0&0\\0&0&0&0\end{bmatrix}\delta_{2\pi(t)}(m)\overline{\delta_{2\pi(t)+1}(m)}
z^{2^{2\pi(t)}}+ \nonumber\\
&&\begin{bmatrix}0&0&0&0\\0&0&0&0\\0&0&1&0\\0&0&0&0\end{bmatrix}\overline{\delta_{2\pi(t)}(m)}\delta_{2\pi(t)+1}(m)
z^{2\cdot 2^{2\pi(t)}}
+\left.\begin{bmatrix}0&0&0&0\\0&0&0&0\\0&0&0&0\\0&0&0&1\end{bmatrix}\delta_{2\pi(t)}(m)\delta_{2\pi(t)+1}(m)
z^{3\cdot2^{2\pi(t)}} \right)
H^{\{t+1\}}\nonumber.
\end{eqnarray}
\end{footnotesize}

Define
\begin{footnotesize}
\begin{eqnarray}
&&F_{\pi(t)}(\delta_{2\pi(t)}(m),\delta_{2\pi(t)+1}(m))
\nonumber\\
&=&\Big(\begin{bmatrix}1&0&0&0\\0&0&0&0\\0&0&0&0\\0&0&0&0\end{bmatrix}\cdot\overline{\delta_{2\pi(t)}(m)}\cdot
\overline{\delta_{2\pi(t)+1}(m)}+\begin{bmatrix}0&0&0&0\\0&1&0&0\\0&0&0&0\\0&0&0&0\end{bmatrix}\cdot\delta_{2\pi(t)}(m)\cdot\overline{\delta_{2\pi(t)+1}(m)}\cdot
z^{2^{2\pi(t)}}+\nonumber\\
&&\begin{bmatrix}0&0&0&0\\0&0&0&0\\0&0&1&0\\0&0&0&0\end{bmatrix}\overline{\delta_{2\pi(t)}(m)}\delta_{2\pi(t)+1}(m)\cdot
z^{2^{2\cdot 2\pi(t)}}+\begin{bmatrix}0&0&0&0\\0&0&0&0\\0&0&0&0\\0&0&0&1\end{bmatrix}\delta_{2\pi(t)}(m)\delta_{2\pi(t)+1}(m)
z^{3\cdot 2^{2\pi(t)}}\Big)\cdot H^{\{t+1\}}.\nonumber
\end{eqnarray}
\end{footnotesize}
Then it is easy to see that \begin{footnotesize}
$$F_{\pi(t)}(0,0)=\begin{bmatrix}1&0&0&0\\0&0&0&0\\0&0&0&0\\0&0&0&0\end{bmatrix}\cdot
H^{\{t+1\}},$$
$$F_{\pi(t)}(1,0)=\begin{bmatrix}0&0&0&0\\0&1&0&0\\0&0&0&0\\0&0&0&0\end{bmatrix}\cdot z^{2^{2\pi(t)}}\cdot
H^{\{t+1\}},$$
$$F_{\pi(t)}(0,1)=\begin{bmatrix}0&0&0&0\\0&0&0&0\\0&0&1&0\\0&0&0&0\end{bmatrix}\cdot
z^{2\cdot 2^{2\pi(t)}}\cdot H^{\{t+1\}},$$
and
$$F_{\pi(t)}(1,1)=\begin{bmatrix}0&0&0&0\\0&0&0&0\\0&0&0&0\\0&0&0&1\end{bmatrix}\cdot
z^{3\cdot 2^{2\pi(t)}}\cdot H^{\{t+1\}}.$$
\end{footnotesize}

By ysing Corollary 4,  we have
\begin{footnotesize}
\begin{eqnarray}
& &\sum^{L-1}_{m=0} H^{\{0\}}\cdot\prod^{n}_{t=1}\left(A_t\cdot
H^{\{t\}}\right)z^m \nonumber\\
&=& H^{\{0\}}\cdot\left\{\prod_{t=0}^{n-1}\left(\begin{bmatrix}1&0&0&0\\0&0&0&0\\0&0&0&0\\0&0&0&0\end{bmatrix}\cdot
H^{\{t+1\}}\right.\right.+\begin{bmatrix}0&0&0&0\\0&1&0&0\\0&0&0&0\\0&0&0&0\end{bmatrix}\cdot
z^{2^{2\pi(t)}}\cdot H^{\{t+1\}}+\nonumber\\
&&\begin{bmatrix}0&0&0&0\\0&0&0&0\\0&0&1&0\\0&0&0&0\end{bmatrix}\cdot
z^{2\cdot 2^{2\pi(t)}}\cdot H^{\{t+1\}}+\left.\left.\begin{bmatrix}0&0&0&0\\0&0&0&0\\0&0&0&0\\0&0&0&1\end{bmatrix}\cdot
z^{3\cdot 2^{2\pi(t)})}\cdot H^{\{t+1\}}\right)\right\}\nonumber\\
&=&H^{\{0\}}\left\{\prod_{t=1}^{n}\left(\begin{bmatrix}1&0&0&0\\0&z^{4^{\pi(t)}}&0&0
\\0&0&z^{2\cdot4^{\pi(t)}}&0\\0&0&0&z^{3\cdot4^{\pi(t)}}\end{bmatrix}
\cdot H^{\{t\}}\right)\right\}\nonumber\\
&=&M^{\{n\}}(z).
\end{eqnarray}
\end{footnotesize}
This completes the proof.
\end{proof}

Since each $A_t$ has only one nonzero element, only one row of product $A_t\cdot
H^{\{t\}}$ is nonzero. Therefore, every entry of  $M^{\{n\}}(m)$ is
a product of the specific entries of $H^{\{t\}}$ for $t\in{0,1,2,\cdots,n},$ which must be a $q$th root of unity.
\begin{example}
In Construction \ref{cons:main}, let
$H^{\{0\}}=\begin{bmatrix}1&1&1&1\\1&i&-1&-i\\1&-i&-1&i\\1&-1&1&-1\end{bmatrix}$,
$H^{\{1\}}=\begin{bmatrix}1&i&-1&-i\\1&1&1&1\\1&-1&1&-1\\1&-i&-1&i\end{bmatrix}$ and
$H^{\{2\}}=\begin{bmatrix}1&1&1&1\\1&i&-1&-i\\1&-1&1&-1\\1&-i&-1&i\end{bmatrix}$
be Butson type $(4,4)$ Hadamard matrices. The Boolean functions
corresponding to the elements of $M^{\{3\}}(z)=H^{\{0\}}\cdot
D(z)^{4^1}\cdot H^{\{1\}}\cdot D(z)^{4^0}\cdot H^{{\{2\}}}$ were calculated as
follows (Suppose $Boolfunc_{r,s}=f_{r,s}(x_0, x_1, x_2, x_3)$ be the
Boolean function corresponding to $M_{r,s}^{\{3\}}(z)$, where $x_0$
is the most significant bit).
\begin{flushleft}
$Boolfunc_{0,0}=3 x_0 + x_1 + 2x_2 + x_3 + 2x_0x_1 + 2x_0x_2 +
3x_0x_3 + 2x_1x_2 + x_1x_3 + 2x_0x_1x_3 + 1$

$ Boolfunc_{0,1}=3x_0 + x_1 + 2x_0x_1 + 2x_0x_2 + 3x_0x_3 + 2x_1x_2
+ x_1x_3 + 2x_0x_1x_3$

$Boolfunc_{0,2}=3x_0 + x_1 + 2x_2 + 3x_3 + 2x_0x_1 + 2x_0x_2 +
3x_0x_3 + 2x_1x_2 + x_1x_3 + 2x_0x_1x_3 + 3$

$Boolfunc_{0,3}=3x_0 + x_1 + 2x_3 + 2x_0x_1 + 2x_0x_2 + 3x_0x_3 +
2x_1x_2 + x_1x_3 + 2x_0x_1x_3 + 2$

$Boolfunc_{1,0}=x_0 + 2x_2 + x_3 + 2x_0x_1 + 2x_0x_2 + 3x_0x_3 +
2x_1x_2 + x_1x_3 + 2x_0x_1x_3$

$Boolfunc_{1,1}=x_0 + 2x_0x_1 + 2x_0x_2 + 3x_0x_3 + 2x_1x_2 + x_1x_3
+ 2x_0x_1x_3 + 3$

$Boolfunc_{1,2}=x_0 + 2x_2 + 3x_3 + 2x_0x_1 + 2x_0x_2 + 3x_0x_3 +
2x_1x_2 + x_1x_3 + 2x_0x_1x_3 + 2$

$Boolfunc_{1,3}=x_0 + 2x_3 + 2x_0x_1 + 2x_0x_2 + 3x_0x_3 + 2x_1x_2 +
x_1x_3 + 2x_0x_1x_3 + 1$

$Boolfunc_{2,0}=x_0 + 2x_3 + 2x_0x_1 + 2x_0x_2 + 3x_0x_3 + 2x_1x_2 +
x_1x_3 + 2x_0x_1x_3 + 1$

$Boolfunc_{2,1}=x_0 + 2x_1 + 2x_0x_1 + 2x_0x_2 + 3x_0x_3 + 2x_1x_2 +
x_1x_3 + 2x_0x_1x_3 + 1$

$Boolfunc_{2,2}=x_0 + 2x_1 + 2x_2 + 3x_3 + 2x_0x_1 + 2x_0x_2 +
3x_0x_3 + 2x_1x_2 + x_1x_3 + 2x_0x_1x_3$

$Boolfunc_{2,3}=x_0 + 2x_1 + 2x_3 + 2x_0x_1 + 2x_0x_2 + 3x_0x_3 +
2x_1x_2 + x_1x_3 + 2x_0x_1x_3 + 3$

$Boolfunc_{3,0}=3x_0 + 3x_1 + 2x_2 + x_3 + 2x_0x_1 + 2x_0x_2 +
3x_0x_3 + 2x_1x_2 + x_1x_3 + 2x_0x_1x_3 + 3$

$Boolfunc_{3,1}=3x_0 + 3x_1 + 2x_0x_1 + 2x_0x_2 + 3x_0x_3 + 2x_1x_2
+ x_1x_3 + 2x_0x_1x_3 + 2$

$Boolfunc_{3,2}=3x_0 + 3x_1 + 2x_2 + 3x_3 + 2x_0x_1 + 2x_0x_2 +
3x_0x_3 + 2x_1x_2 + x_1x_3 + 2x_0x_1x_3 + 1$

$Boolfunc_{3,3}=3x_0 + 3x_1 + 2x_3 + 2x_0x_1 + 2x_0x_2 + 3x_0x_3 +
2x_1x_2 + x_1x_3 + 2x_0x_1x_3$

\end{flushleft}
\end{example}
In the above example,  any  set $  \{Boolfunc_{i,j}|j=0,1,2,3\}$ or $  \{Boolfunc_{i,j}|i=0,1,2,3\}$ is a complementary set of size 4.
It can be seen that any  sequences in the complementary set   have  algebraic degree 3.
These complementary set cannot be obtained by the known methods given in \cite{paterson} and \cite{schmidt07}.

\begin{example}
In Construction \ref{cons:main}, let $N=2$ and choose
$H^{\{t\}}=\begin{bmatrix}1&\theta^{c_t}\\ -\theta^{c_t}&1\end{bmatrix}$, where $\theta$ is a primitive $q$th root and $c_t\in\{0,1,\cdots, q-1\}$. Then all the $q$-ary  standard Golay pairs can be presented as the first row of $M^{\{n\}}(z)$. This results have been proved by Bud\v{i}sin and  Spasojevi$\acute{c}$'s  in \cite{BudIT}.
\end{example}

\section{Conclusion}

Inspired by the paraunitary-matrix-based method in \cite{BudIT}, we propose a new method to construct $q$-ary  Golay complementary set of size $N$ and length $N^n$ by $N\times N$ Hadamard matrices where $n$ is arbitrary and $N$ is a power of 2. Each item of the constructed sequences can be presented as the product of the specific entries of the Hadamard matrices. To make the results easier to understand for readers, we only show the proof for $N=4$. Bud\v{i}sin and  Spasojevi$\acute{c}$'s algorithm for $q$-ary Golay pair in \cite{BudIT} can be regarded as a special case of the construction in this paper by setting $N=2$ and choosing specific Hadamard matrices. We also show new quaternary Golay pair of size 4 and degree 3 never reported in the literature can be obtained by the method presented in this paper.

Since any Butson type Hadamard  matrix can be used in our construction, and in general, matrix multiplication is not commutative, a large families of
sequences with PMEPR upper bounded by constant can be obtained by choosing different Hadamard matrix in Construction \ref{cons:main}.
The details will be considered in the extended  paper.


%



\end{document}